%% file: main.tex
\documentclass[letterpaper, 12pt]{article}
\usepackage{fullpage}
\usepackage{graphicx}
\usepackage[round]{natbib}
\usepackage{hyperref}
\usepackage[lined,boxed]{algorithm2e}

\providecommand{\DontPrintSemicolon}{\dontprintsemicolon}
\usepackage{rfmacros}
\usepackage{swmacros}

\newcommand\bigTheta[1]{\Theta\left( #1 \right)}
\newcommand\littleO[1]{o\left( #1 \right)}

\newcommand\loc[1]{\widehat{{#1}}}

\newcommand\IQP[1]{\operatorname{QP}_{#1}}
\newcommand\SDP[1]{\operatorname{SDP}_{#1}}

\renewcommand\symmat[1]{\mathbb{S}_{#1}}  % Set of #1 by #1 symmetric matrices
\newcommand\sphere[1]{\mc S^{#1}}  % Unit sphere in ambient dimension #1
\newcommand\maxeig{\lambda_{\text{max}}}
\newcommand\epnet{\epsilon\operatorname{-net}}

\renewcommand\todo[1]{}
\renewcommand\todomath[1]{}

\title{A sub-constant improvement in approximating the
  positive semidefinite Grothendieck problem}

\date{}
\author{
  Roy Frostig \\
  Stanford University \\
  \texttt{rf@cs.stanford.edu}
  \and
  Sida I.\ Wang \\
  Stanford University \\
  \texttt{sidaw@cs.stanford.edu}}

\begin{document}

\maketitle

\begin{abstract}
  Semidefinite relaxations are a powerful tool for approximately
  solving combinatorial optimization problems such as MAX-CUT and the
  Grothendieck problem.
  By exploiting a bounded rank property of extreme points in the
  semidefinite cone, we make a sub-constant improvement in the
  approximation ratio of one such problem.
  Precisely, we describe a polynomial-time algorithm for the positive
  semidefinite Grothendieck problem -- based on rounding from the
  standard relaxation -- that achieves a ratio of $2/\pi +
  \Theta(1/{\sqrt n})$, whereas the previous best is $2/\pi +
  \Theta(1/n)$.
  We further show a corresponding integrality gap of
  $2/\pi+\tilde{O}(1/n^{1/3})$.
\end{abstract}

\section{Introduction}
\label{sec:intro}

% The problem
Given a positive semidefinite (PSD) matrix $A \in \R^{n \by n}$, the
positive semidefinite Grothendieck problem is
\begin{align}
  \max_{x \in \hypc^n} \IQP{A}(x) \defeq x^\T A x.
  \label{eqn:iqp}
\end{align}
The problem is NP-hard; it is easy to see that MAX-CUT arises as the
special case when $A$ is a graph Laplacian. Elsewhere, the problem has
applications ranging from graph partitioning
\citep{alon2006approximating} to kernel clustering
\citep{khot08approximate, khot10sharp}. See
\citet{pisier12grothendieck} for a broad survey.

% The algorithm family
The polynomial-time algorithm that achieves the asymptotically best
known approximation ratio for this problem -- the constant $2/\pi
\approx 0.637$ -- is essentially the same as that described by
\citet{goemans1995improved} for MAX-CUT: the problem \eqnref{eqn:iqp}
is relaxed to a convex semidefinite program (SDP) that is equivalent
to
\begin{align}
  \maxproblem {X \in \R^{n \by n}} {\tr(X^\T A X)} {& \norm{X_i}_2^2 =
    1}
  \label{eqn:sdpxx}
\end{align}
where $\set{X_i}$ are the rows of $X$. We can think of each $X_i$ as comprising a relaxation of
the binary variable $x_i$ to an $n$-dimensional real unit
vector. This convex relaxation is solved to arbitrary accuracy and
its solution randomly rounded to a discrete one for
\eqnref{eqn:iqp}.\footnote{Though the algorithms are the same, the
  MAX-CUT approximation constant (about $0.878$) exceeds $2/\pi$ by an
  analysis that exploits all-positive edge weights and the graph
  Laplacian structure of $A$.}
% Rong points out that the following is unnecessary for this audience.
%% Rounding is a randomized procedure and the objective value of the
%% rounded solution may be lower than that of its corresponding relaxed
%% point. A bound on the expected ratio of objectives is precisely what
%% comprises the constant approximation factor.

% The reason to believe we are at a barrier
There is evidence suggesting that this approximation is asymptotically
optimal. In particular, \citet{alon2006approximating} exhibit a random
problem instance $A$ whose asymptotic integrality gap is
$2/\pi$.
For the SDP relaxation approach, the integrality gap bounds the
approximation ratio from above.
In general, \citet{khot08approximate} show that if the unique games conjecture
holds then no polynomial time algorithm can exceed a $2/\pi$ ratio
guarantee.

% A new frontier: sub-constant improvements
However, we can still push further against this barrier. To do so, we
look to approximation ratios that, for any \emph{fixed $n$}, exceed an
asymptotic limit of $2/\pi$, and to favor those algorithms whose ratio
decays \emph{more slowly} to the asymptotic limit. Such an algorithm
is said to provide a \emph{sub-constant} improvement to the
$2/\pi$-approximation.  It is the best kind of improvement possible
that avoids confrontation with asymptotic hardness barriers.

% Previous best on this front. Intuition: really 1/k but k=n
The initial proof that the SDP rounding procedure used in
\citet{goemans1995improved} for MAX-CUT can be repurposed for a
$2/\pi$ approximation to the PSD Grothendieck problem is due to
\citet{nesterov1998semidefinite}. More recently,
\citet{briet2010positive} showed, by a careful analysis, that the
procedure in fact achieves an approximation ratio of
$2/\pi+\Theta(1/n)$. Intuitively, as the problem instance grows, the
dimension of the relaxed variables $\set{X_i}$ grows with it, and the
expected gain of rounding (over $2/\pi$) decreases inversely with the
relaxed dimension.

% Our contribution:
% - FIRST INSIGHT: lower effective dimension is better for rounding
% - (intermediate) So why don't we try to reduce dimensionality? Maybe we can do it without much objective cost, so we gain overall
% - SECOND INSIGHT: we in fact have a huge dim reduction FOR FREE
% - THIRD INSIGHT: we get *further* dim reduction on easily-characterizable problem instances
The key to our improvement is twofold.
First, the expected gain is larger when the relaxed variables
$\set{X_i}$ all lie in a low-dimensional subspace of $\R^n$. This
leads us to seek a polynomial-time dimensionality reduction that
improves the expected gain more than it decreases the SDP objective
value.
Second, by controlling the matrix rank of optimal solutions to
semidefinite programs, we can actually obtain an immediate such
dimensionality reduction -- down to below $\sqrt{2n}$ -- at
\emph{entirely no cost in objective}. To our knowledge, this is the
first use of such an essential property of SDP extreme points in the
context of approximation by SDP relaxation.
Lastly, our analysis can be sharpened for problem instances in which
$A$ itself is low rank. In these cases, we show an approximation ratio
of $2/\pi + \Theta(1/\rank(A))$. In particular, this implies a
characterization of problem instances -- those with $A$ of constant
rank -- for which the algorithm we present achieves a \emph{constant}
improvement.

\subsection{Formal setup and main result}
\label{sec:setup}

\paragraph{Notation}
We write $\symmat k$ for the set of symmetric $k \by k$ real matrices
and $\sphere{k}$ for the unit sphere $\{ x \in \R^k : \| x \|_2 = 1 \}$.
%and $\vec 1_q$ for the all-ones length-$q$ vector.
All vectors are columns unless stated otherwise.
If $X$ is a matrix, then $X_i \in \R^{1 \by k}$ is its $i$'th row.

\

\noindent We are interested in approximately solving the positive
semidefinite Grothendieck problem \eqnref{eqn:iqp} by rounding an
optimal solution of the relaxed problem \eqnref{eqn:sdpxx}.  As
stated, \eqnref{eqn:sdpxx} is not convex, but it does correspond
exactly to a (convex) semidefinite program through the change of
variables $S=XX^\T$:
\begin{align}
  \maxproblem { S \in \symmat{n} }
              { \SDP{A}(S) \defeq \tr(A S) }
              {
                & S \succeq 0 , \;\;
                  \diag(S) = \vec 1
              }.
  \label{eqn:sdp}
\end{align}
Note that if the rank of some feasible $S$ equals $k$, then the
corresponding $X$ feasible for \eqnref{eqn:sdpxx} has rows that are
effectively $k$-dimensional.

Because $\SDP{A}$ is a convex program, we can obtain an optimal
solution $\opt S$ of $\SDP{A}$ within a desired precision $\ep$ in
time polynomial in $n$ and $\log(1/\ep)$. From an optimal SDP point
$\opt S$, we can obtain a feasible point $\loc x \in \hypc^n$ for
$\IQP{A}$ by the following randomized rounding procedure: factor $S$
into $XX^\T$, sample a random vector $g$ from the unit sphere
$\sphere{k}$, and output $\loc x = \sign(X g)$.

This randomized rounding is analyzed independently by
\citet{goemans1995improved} and by
\citet{nesterov1998semidefinite}. Both show that the
approximation ratio is bounded above and below as follows:
\begin{equation}
\frac 2 \pi \leq \frac{ \E[\IQP{A}(\loc x)] } { \SDP{A}(\opt S) } \leq 1
\end{equation}
and, as mentioned above, this is the asymptotically optimal
approximation ratio of any polynomial-time algorithm, provided that
the unique games conjecture holds.

Adapting the rounding analysis of \citet{briet2010positive} and
controlling the rank of SDP solutions, we obtain in this paper an
approximation ratio of
\begin{align}
  \frac{2}{\pi} + \frac 1 {\pi\sqrt{2n}} + \littleO{ \frac1{\sqrt{n}} }
  &= \frac{2}{\pi} + \bigTheta{\frac 1 {\sqrt n}}.
  \label{eqn:our_ratio}
\end{align}

% The result is achieved by bridging two analyses.  The first concerns
% rounding SDP solutions from finite rank. The second concerns
% rank-reduction of SDP solutions.
Section~\ref{sec:sdprank} shows that solutions of $\SDP{A}$ with low
rank -- bounded above by $\sqrt{2n}$ -- always exist, and describes a
polynomial-time algorithm for finding them.
% low rank reference
Section~\ref{sec:roundfromk} shows the approximation ratio achieved by
the randomized rounding algorithm applied to a $k$-rank solution of
$\SDP{A}$ for a known $k \leq n$.  Combining these results via
$k=\sqrt{2n}$ yields the main result \eqnref{eqn:our_ratio}. Finally,
Section~\ref{sec:tightness} adapts the analysis of
\citet{alon2006approximating} to show a corresponding upper bound of
the integrality gap -- and hence the best approximation
guarantee possible via the SDP rounding approach -- is at most
$2/\pi+\tilde{O}(1/n^{1/3})$.

In addition to the main result, Section~\ref{sec:lowrankA} adapts the
rank reduction algorithm in \ref{sec:lowrankA} to further improve the
approximation ratio whenever $A$ has rank $o(\sqrt n)$.
% , and in particular to provide a constant improvement -- even as $n
% \to \infty$ -- when $A$ has fixed rank.

% Any problem-specific features of $A$, or general approximations schemes
% that can strengthen either section, would yield immediate further
% improvements.

\section{SDP solution rank}
\label{sec:sdprank}

Considering only the constraint count of a semidefinite program, while
ignoring its objective altogether, \citet{barvinok95problems} and
\citet{pataki1998rank} argue geometrically that SDP solutions have
bounded rank:

\begin{thm}[\citet{barvinok95problems, pataki1998rank}] Any
  semidefinite program of $m$ linear constraints has an optimal
  solution $\opt S$ such that $t(\rank(\opt S)) \leq m$, where $t(k) =
  k(k+1)/2$ is the $k$'th triangular number.
  \label{thm:sqrt2n}
\end{thm}

Since $\SDP{A}$ has only $n$ constraints -- those of the form $S_{ii}
= 1$ -- it follows that it has an optimal solution whose rank does not
exceed roughly $\sqrt{2n}$.

\citet{alfakih98embeddability} give a concrete algorithm for finding
the low-dimensional Euclidian embeddings shown to exist in the proof
of \citet{barvinok95problems}. The algorithm is essentially a
constructive version of the existence proof concurrently given by
\citet{pataki1998rank}.
By simplifying their key ideas and translating them to the problem of
rank-reducing solutions of $\SDP{A}$, we obtain
Algorithm~\ref{alg:rankred}, which reduces the rank of an SDP solution
$S$ without changing its objective value nor affecting feasibility.
The algorithm proceeds by solving a homogeneous linear system that is
underdetermined whenever $\rank(S)$ is sufficiently large.

\begin{algorithm}[h!]
  \DontPrintSemicolon

  \SetKwInOut{Input}{Input}\SetKwInOut{Output}{Output}
  \Input{SDP solution $S \in \symmat n$ of rank $k$, with $t(k) > n+1$.}
  \Output{SDP solution $S'$ of rank $k'$, with $t(k') \leq n+1$.}

  \BlankLine
  Note that $t(k) = k(k+1)/2$ is the dimension of $\symmat{k}$.

  \BlankLine
  Factor $S = X X^\T$ with $X \in \R^{n \by k}$. \;

  Solve $n+1$ homogeneous linear equations in $t(k)$ variables $Y \in \symmat{n}$: \;
  \quad\quad $\tr(X_i^\T X_i\ Y) = 0$ for each of $X$'s $n$ row vectors $X_i \in \R^{1 \by k}$ \;
  \quad\quad $\tr(X^\T A X\ Y) = 0$

  \BlankLine  
  Negate and scale $Y\neq 0$ if needed, so that its largest eigenvalue $\maxeig=1$. \;
  Set $U  \gets I_k - Y$ and $S' \gets X U X^\T$.

  \caption{Rank-reduction of an $\SDP{A}$ solution}
\label{alg:rankred}
\end{algorithm}

To see that Algorithm~\ref{alg:rankred} delivers on its promises, observe that $\rank(S')
\leq \rank(U) < \rank(S)=k$ because $\det(U)=\det(Y - \maxeig I_k) =
0$ for the eigenvalue $\maxeig=1$ of $Y$. We can check that
$U\succeq 0$ and therefore $S' \succeq 0$. Because we found $Y$
satisfying the linear system, we can also check that as far the constraints
and objective of $\SDP{A}$ are concerned, $S'$ is no worse than $S$.
The resulting objective value is
\begin{align}
\label{eq:objnochange}
\nonumber
\tr (A S') &= \tr(A X U X^\T) = \tr(X^\T A X U) =\tr(X^\T A X (I_k -
Y)) \\ 
& = \tr(X^\T A X I_k) = \tr(A X X^\T ) = \tr(A S).
\end{align}
Similarly, the new solution remains feasible:
\begin{align}
  S'_{ii} &= X_i U X^\T_i = \tr(X^\T_i X_i U) = \tr(X^\T_i X_i I_k) =
  S_{ii}=1.
\end{align}

\subsection{Low rank problem instances}
\label{sec:lowrankA}
Further rank-reduction is possible for problem instances with
additional structure. In this section we show that, when $A$ is low
rank, it is possible to modify Algorithm~\ref{alg:rankred} so that it
reduces solution rank to the rank of $A$.
% That is, when combined with the coming sections, we will have that the
% approximation ratio presented as the main result
% \eqnref{eqn:our_ratio} can be more specifically written as,
% \begin{align}
%   \frac{2}{\pi} + \bigTheta{\frac 1 {\min \set{ \sqrt n, \rank(A) }}}.
%   \label{eqn:our_ratio_full}
% \end{align}

To exploit the rank of $A$, we replace the linear homogeneous
equations in Algorithm~\ref{alg:rankred} with the semidefinite program,
\begin{equation}
\label{eq:posdefmodifier}
  Y \succeq 0,\ \tr(X^\T A X\ Y) = 0,\  Y \neq 0,
\end{equation} 
and claim that is it feasible whenever $k > \rank(A)$.
To see this, diagonalize $X^\T A X = Q \diag(\lambda) Q^\T$ with
orthonormal eigenvectors $Q \in \R^{n \by k}$ and eigenvalues $\lambda
\in \R^k$.
Since $k > \rank(A)$, there exists $i$ such that $\lambda_i=0$. If we
assign the non-zero vector $\lambda' \in \R^k$ as
\begin{align}
\lambda'_j \gets \indic [\lambda_j = 0],\  j = 1, \ldots, k, 
\end{align}
then $Y = Q \diag(\lambda') Q^\T$ satisfies \eqref{eq:posdefmodifier}. 

After solving for $Y$, the last step of Algorithm~\ref{alg:rankred}
computes the rank-reduced solution $S'$.
% maybe say how to solve this: just get a diagonalization of A=P\diag(d)P',
% make f=0; f(d==0) = 1, and f(d~=0) = 0.
We can follow \eqref{eq:objnochange} to check that $S'$ gives us the
same objective value. For feasibility, we have
\begin{align*}
  S'_{ii} &= X_i U X^\T_i =  X_i (I_k-Y) X^\T_i = S_{ii} - X_i Y
  X^\T_i = 1 - X_i Y X^\T_i.
\end{align*}
Recall $\maxeig(Y)=1$, $Y \succeq 0$, and $||X_i||_2 = 1$. Therefore,
$0 \leq S'_{ii} \leq 1$.\footnote{Although the constraint $S_{ii} = 1$
  appears in the formal problem setup, the constraint $S_{ii} \leq 1$
  is equivalent for the PSD Grothendieck problem due to having $A
  \succeq 0$.}
 
% TODO (sidaw): perhaps cite self later, and say this is an
% alternative proof

% TODO (rf): Commented out the paragraph below only because said
% diag(S) = 1 when defining SDP{A}, so this is confusing.  Would be
% nice to have it back if we can avoid that.

% In some particular problem instances, $Y$ can be found by means other
% than solving the linear system.  For example, if we have know the
% weight matrix $A$ rank below $k$, we may instead seek $Y \succeq 0$
% such that $\tr(X^\T A X Y) = 0$; a non-trivial such $Y$ exists when
% $\rank(X^\T A X) \leq \rank(A) < k$.  In that case, our resulting
% solution $S'$ has rank at most that of $A$.  This problem-specific
% procedure arises from the fact that it suffices to have $S'_{ii} \leq
% 1$.

\section{Rounding from low rank}
\label{sec:roundfromk}
The following lemma states that the approximation ratio due to
randomized rounding is better when rounding from lower-rank SDP
solutions. The statement is a simple consequence of Lemma 1 of
\citet{briet2010positive}, which makes important use of the results of
\citet{schoenberg1942positive} together with Grothendieck's identity.
\begin{lem} 
  Fix a weight matrix $A \succeq 0$ and $X \in \R^{n \by k}$ with $X_i
  \in \sphere k$, the unit sphere. Let $g$ be a random vector from
  $\sphere{k}$ and
  \begin{align}
  \gamma(k)
    \eqdef \frac 2 k \left( \frac { \Gamma((k+1)/2) } { \Gamma(k/2) }
    \right)^2 = 1 - \bigTheta{\frac1{k}}.
  \end{align}
  Then the expected approximation ratio obtained by randomized rounding
  \begin{align}
    R(k)
    &\eqdef
      \frac{\E_g[\IQP{A}(\sgn(X g))]}{\SDP{A}(X X^\T)}
  \end{align}
  is at least
  \begin{align}
    \frac{2}{\pi\gamma(k)}
    &=
      \frac{2}{\pi} \left( 1+ \frac{1}{2k} + \littleO{ \frac 1 k } \right).
  \end{align}
\label{lem:roundingratio}
\end{lem}

\begin{proof}
Grothendieck's identity states that, for $u, v \in \R^k$ and $g$
drawn uniformly from the unit sphere $\sphere k$,
\begin{align}
\E_g \left[ \sgn(u^\T g) \sgn(v^\T g) \right] = \frac 2 \pi \arcsin(u^\T v).
\end{align}

Let $Y = f(X X^\T) \in \R^{n \by n}$ be the elementwise application of
the scalar function
\begin{align}
  f(t) = \tfrac{2}{\pi} \left( \arcsin(t) - \tfrac{t}{\gamma(k)} \right).
\end{align}
Lemma 1 in \citet{briet2010positive} shows that $f(t)$ is a function
of the \emph{positive type} on $\sphere k$, which by definition means
that $Y \succeq 0$ provided $X_i \in \sphere k$ for all $i$. Their
result is based on
(a) computing inner products between orthogonal
Jacobi polynomials, together with
(b) the characterization due to \citet{schoenberg1942positive} of positive definite functions on $\sphere k$ in terms of Jacobi polynomials.

We have that $\tr(AY) \geq 0$.  Rearranging terms and applying
Grothedieck's identity:
\begin{align}
  & 0 \leq \tr(AY)
    = \tr \left(A \frac{2}{\pi} \left( \arcsin(X X^\T) - \frac{X X^\T}{\gamma(k)} \right) \right) \\
  & \iff  \tr \left( A \frac{2}{\pi} \arcsin(X X^\T) \right) \geq \frac{2}{\pi \gamma(k)} \tr(A X X^\T) \\
  & \iff  \E_g[\IQP{A}(\sgn(Xg))] \geq \frac{2}{\pi \gamma(k)} \SDP{A}(XX^\T),
\end{align}
which proves the claim.
\end{proof}

\input{tightness}

\section{Concluding remarks}
\label{sec:conc}

We demonstrated a sub-constant improvement
in approximating the PSD Grothendieck problem. Although the
improvement disappears asymptotically, it decays slowly via an
additive term whose constant factors we have made explicit.
% In a sense, this result provides a ``proof of concept'' in making
% the sort of algorithmic improvement that does not altogether
% overturn long-conjectured asymptotic hardness barriers.
Two of the three main ingredients of this result are obtained by
adapting existing analyses to explicitly account for effective relaxed
dimension in the ``first order'' sub-constant additive term. The
remaining ingredient comes from exploiting the spectral sparsity of
extreme points in the SDP cone, an analysis tool of independent
interest.
With the same tool set, we further characterized a class of problem
instances for which the new approximation ratio enjoys an additional
-- even constant -- advantage.

An immediate direction for future work is to ask whether the
sub-constant improvement described here has downstream
implications for other approximation algorithms. Another is whether the result
can be improved, or conversely whether the integrality gap is actually
smaller than shown in Section~\ref{sec:tightness}. A more general
question is whether this same set of tools can be applied to other
SDP relaxation-based algorithms in order to improve their approximation ratio -- by
an additive sub-constant term or otherwise -- with immediate
candidates being MAX-CUT, $k$-coloring, and kernel
clustering.

\bibliographystyle{unsrtnat}
\bibliography{all}

%\appendix
%\input{extensions}

\end{document}

%% file: tightness.tex
\section{Integrality gap}
\label{sec:tightness}

How much further could we hope to improve the additive sub-constant
term in the ratio between rounded and relaxed solutions?  This section
bounds the answer by providing an integrality gap of $2/\pi +
\tilde{O}(1/n^{1/3})$.

To establish the gap, we set out to construct, for every $n$, a matrix $A \in
\R^{n \by n}$ so that
$\frac{\IQP{A}(x^*)} {\SDP{A}(S^*)} \leq 2/\pi + \tilde{O}(1 / n^\alpha )$.
The particular construction we consider achieves $\alpha = 1/3$.
%% As a result, it impossible to
%% improve $\alpha$ to smaller than $1/3$ from our current $1/2$ .
We first outline and reproduce some results from Section~5.2 of
\citet{alon2006approximating}, and then expand them to analyze the
sub-constant rates.

The authors' original construction uses $n$ random unit vectors
$v_i \in \sphere p$ for $i=1,\ldots, n$ and takes $A_{ij} = \frac1{n} v_i^\T v_j$. If we set $S_{ij}=A_{ij}$ then
\begin{align}
\label{eq:sdplowerbound}
\SDP{A}(S^*) \geq \SDP{A}(S) = \tr(AS) =\frac1{n^2} \sum_{ij} (v_i^\T v_j)^2 \to 1/p,
\end{align}
where $1/p$ arises as the average inner product between random vectors
on $\sphere p$.

Under the QP, for any $x \in \hypc^n$, we have
\begin{align}
\IQP{A}(x) = \sum_{i,j}^n A_{ij} x_i x_j = \norm{ \frac1{n} \sum_{i=1}^n
v_i x_i }^2.
\end{align}
%(consider filling in more details from alon/naor)
Take $x^* \in \arg\max_x \IQP{A}(x)$ and let $c$ be the unit vector the direction of
$\sum_{i=1}^n x^*_i v_i$. It is optimal to accumulate in
the correct direction $c$, so $x^*_i = \sgn(v_i^\T c)$ and hence
\begin{align}
\IQP{A} (x^*) = \left(\frac1{n} \sum_{i=1}^n x^*_i v_i^\T c\right)^2 =
\left(\frac1{n} \sum_{i=1}^n \abs{v_i^\T c} \right)^2 \to ( \E \bigsqbra{ \abs{v^\T c} })^2.
\end{align}
% reproduce more details here or just cite?
\citet{alon2006approximating} computed this expectation; it is easy to
verify that the sub-constant term $\Theta(1/p)$ appears therein
as follows:
\begin{align}
\label{eq:expabs}
\E\bigsqbra{ \abs{v^\T c} } &= \left(\sqrt{\frac{2}{\pi}} +
  \Theta\left(\frac1 p\right) \right) \frac1{\sqrt{p}} .
\end{align}

We would now like to maximize an $n$-sample estimate of
\eqnref{eq:expabs} over the sphere. The original analysis does this by
replacing maximization over the sphere with the same over a
corresponding $\ep$-net:
\begin{align}
  \label{eq:maxoversample}
\IQP{A} (x^*) = \left( \max_{d \in \sphere p}
  \sum_{i=1}^n  \frac1{n}   \abs{v_i^\T d} \right)^2 = \left( \left( \max_{d \in \epnet(\sphere p)}
  \sum_{i=1}^n \frac1{n}   \abs{v_i^\T d} \right) + O(\epsilon) \right)^2.
\end{align}
Now, $n$ needs to be big enough so the variance of the $n$-sample estimator,
\begin{align}
\label{eq:variance}
\Var\left[\sum_{i=1}^n \frac1{n}   \abs{v_i^\T d} \right] = \frac1{n}
\Var\bigsqbra{ \abs{v_i^\T d} } = O\left(\frac1{np}\right),
\end{align}
is small enough to safely maximize over an $\epnet$ of size
$O(\frac1{\epsilon^p})$. The integrality gap question then reduces to
the question of how big $n$ should be.

To handle the max, we observe that $\sum_{i=1}^n \frac1{n} \abs{v_i^\T
  d}$ is sub-Gaussian with parameter $O(1/\sqrt{np})$, so it enjoys the
following bound:
If $X_i \sim \gaussian(0, \sigma^2)$ (or if $X_i$ is sub-Gaussian with parameter $\sigma$) are i.i.d across $i=1,\ldots,m$, then
  $\E[\max_i (X_i)] \leq \sigma \sqrt{2 \log(m)}$. 

Now we proceed to bound $\IQP{A}(x^*)$ from above:
\begin{align}
  \label{eq:iqpupperbound}
\IQP{A} (x^*) &= \left( \left( \max_{d \in \epnet(\sphere p)}
  \sum_{i=1}^n \frac1{n}   \abs{v_i^\T d} \right) +
O(\epsilon)\right)^2 \\
& \leq \left( \left(\sqrt{\frac{2}{\pi}} +
  \Theta\left(\frac1 p\right) \right) \frac1{\sqrt{p}}  + \sqrt{\frac{2
\log(\frac1{\epsilon^p})} {np}} + O(\epsilon) \right)^2.
\end{align}

Pick a small enough $\epsilon$ so that the the $O(\epsilon)$ term may
be ignored.
This can be done because the second additive term only grows as
$\log(1/\epsilon)$, so we can pick $\epsilon = o(1 / (p \sqrt p))$
to enforce that the first additive term is dominant.
Multiply both sides of \eqref{eq:iqpupperbound} by the inequality
$1/\SDP{A}(S^*) \leq p$ shown in \eqref{eq:sdplowerbound}. This
yields:
\begin{align}
  \label{eq:finalbound}
\frac{\IQP{A} (x^*)}{\SDP{A}(S^*)}
& \leq \left( \left(\sqrt{\frac{2}{\pi}} +
  \Theta\left(\frac1 p\right) \right) + \sqrt{\frac{2
p \log(\frac1{\epsilon})} {n}} \right)^2.
\end{align}
In order to balance the two sub-constant additive terms, we can set
$n=p^3$. This construction has intergrality gap less than $2 / \pi +
\tilde{O}(1/{n^{1/3}})$.